\documentclass[11pt]{article}
\usepackage[english]{babel}
\usepackage{microtype}
\usepackage{graphicx}
\usepackage{amsmath,amsthm,amssymb,amsfonts}
\usepackage{thmtools,thm-restate}
\usepackage{etoolbox}
\usepackage[dvipsnames]{xcolor}
\usepackage[a4paper,margin=2.5cm]{geometry}
\usepackage{tikz,pgf,pgfplots,adjustbox}
\usetikzlibrary{math,patterns,positioning,arrows,arrows.meta,graphs,graphs.standard,decorations.pathmorphing,tikzmark,calc,hobby,chains,quotes,shapes.geometric}
\pgfplotsset{compat=newest}
\usepackage{xspace}
\usepackage{calc}
\usepackage{enumitem}
\usepackage{booktabs}
\usepackage{url}
\usepackage{doi}
\usepackage[numbers,square,sort&compress]{natbib}
\usepackage{hyperref}
\hypersetup{
    colorlinks=true,
    linkcolor=black,
    filecolor=black,
    urlcolor=black,
    citecolor=black,
    pdftitle={Distributed binary labeling problems in high-degree graphs},
}
\usepackage[nameinlink,capitalize,noabbrev]{cleveref}
\crefname{case}{Case}{Cases}
\Crefname{case}{Case}{Cases}
\crefname{page}{page}{pages}

\newcommand{\eps}{\varepsilon}

\newcommand{\N}{\mathbb{N}}

\renewcommand{\O}{\mathcal{O}}

\DeclareMathOperator{\degen}{degen}
\DeclareMathOperator{\DEG}{DEG}
\DeclareMathOperator{\ARC}{ARC}
\DeclareMathOperator{\leaves}{leaves}
\DeclareMathOperator{\ext}{ext}
\DeclareMathOperator{\depth}{depth}

\newcommand{\n}{0\xspace}
\newcommand{\p}{1\xspace}
\newcommand{\x}{\makebox[\widthof{0}][c]{$\ast$}}

\theoremstyle{plain}
\newtheorem{theorem}{Theorem}[section]
\newtheorem{lemma}[theorem]{Lemma}
\newtheorem{proposition}[theorem]{Proposition}
\newtheorem{corollary}[theorem]{Corollary}

\theoremstyle{definition}
\newtheorem{definition}[theorem]{Definition}
\newtheorem{question}[theorem]{Question}

\theoremstyle{remark}
\newtheorem{remark}[theorem]{Remark}
\newtheorem{example}[theorem]{Example}

\newenvironment{mycover}
{\list{}{\listparindent 0pt
        \itemindent    \listparindent
        \leftmargin    1cm
        \rightmargin   1cm
        \parsep        2pt}%
    \raggedright
    \item\relax}
{\endlist}

\newenvironment{myabstract}
{\list{}{\listparindent 1.5em%
        \itemindent    \listparindent
        \leftmargin    1cm
        \rightmargin   1cm
        \parsep        0pt}%
    \item\relax}
{\endlist}

\newcommand{\myaff}[1]{\,$\cdot$\, {\small #1}\par\smallskip}
\newcommand{\fakeparagraph}[2]{\par\noindent\textbf{#1}\hspace{1em}#2}

\title{Distributed Binary Labeling Problems in High-Degree Graphs}
\author{Henrik Lievonen \and Timoth\'e Picavet \and Jukka Suomela}

\begin{document}
\begin{titlepage}
\vspace*{15mm}
\begin{mycover}
    {\huge\bfseries Distributed Binary Labeling Problems in High-Degree Graphs\par}

    \bigskip
    \bigskip
    \bigskip
    \textbf{Henrik Lievonen}
    \myaff{Aalto University, Finland}

    \textbf{Timoth\'e Picavet}
    \myaff{LaBRI, Universit\'e de Bordeaux, France}

    \textbf{Jukka Suomela}
    \myaff{Aalto University, Finland}
\end{mycover}
\bigskip

\begin{myabstract}
    \fakeparagraph{Abstract.}
    Balliu et al.\ (DISC 2020) classified the hardness of solving \emph{binary labeling problems} with distributed graph algorithms; in these problems the task is to select a subset of edges in a $2$-colored tree in which white nodes of degree $d$ and black nodes of degree $\delta$ have constraints on the number of selected incident edges. They showed that the deterministic round complexity of any such problem is $\O_{d,\delta}(1)$, $\Theta_{d,\delta}(\log n)$, or $\Theta_{d,\delta}(n)$, or the problem is unsolvable. However, their classification only addresses complexity as a function of $n$; here $\O_{d,\delta}$ hides constants that may depend on parameters $d$ and $\delta$.

    In this work we study the complexity of binary labeling problems as a function of all three parameters: $n$, $d$, and $\delta$. To this end, we introduce the family of \emph{structurally simple} problems, which includes, among others, all binary labeling problems in which cardinality constraints can be represented with a context-free grammar. We classify possible complexities of structurally simple problems. As our main result, we show that if the complexity of a problem falls in the broad class of $\Theta_{d,\delta}(\log n)$, then the complexity for each $d$ and $\delta$ is always either $\Theta(\log_d n)$, $\Theta(\log_\delta n)$, or $\Theta(\log n)$.
    
    To prove our upper bounds, we introduce a new, more aggressive version of the \emph{rake-and-compress technique} that benefits from high-degree nodes.
\end{myabstract}
\end{titlepage}

\section{Introduction}

In this work we take the first steps towards characterizing possible distributed computational complexities of graph problems as a function of two parameters: the number of nodes and the maximum degree of the graph. We study so-called \emph{binary labeling problems} \cite{balliu2019classification}, previously only studied for constant-degree graphs, and we extend their classification so that it is parameterized also by the degrees of the nodes. To do that, we also introduce a new version of the \emph{rake-and-compress technique} \cite{Miller1985} that benefits from high-degree nodes.

\subsection{Broader context}

The key goal in the field of distributed graph algorithms is understanding how fast a given graph problem can be solved in a distributed setting. In this work we focus on the LOCAL model of distributed computing (see \cref{ssec:prelim}): in each round all nodes can exchange messages with each of their neighbors, and the running time of the algorithm is the number of communication rounds until all nodes stop and announce their own part of the solution (say, their own color).

\paragraph{Landscape of distributed complexity.}

While a lot of early work in the field focused on individual problems, modern theory of distributed graph algorithms makes a heavy use of \emph{structural} results that apply to a broad family of graph problems. The best-known example is \emph{locally checkable labeling problems} (LCLs), first introduced by \citet{Naor1995}. These are problems in which the task is to label nodes or edges with labels from some finite set, subject to some local constraints---examples of such problems include vertex coloring, edge coloring, maximal independent set, and maximal matching.

By now, we have a very good understanding of the landscape of possible round complexities that \emph{any} LCL problem may have, in settings like cycles, grids, trees, and general graphs \cite{balliu18lcl-complexity,balliu20almost-global,Chang2019,Ghaffari2018,balliu20lcl-randomness,fischer17sublogarithmic,Rozhon2019,brandt16lll,chang2019exponential,ghaffari17distributed}. For example, there is no LCL problem with a time complexity between $\omega(\log^* n)$ and $o(\log n)$ in the deterministic LOCAL model. Gap result like this are powerful tools in proving lower bounds and upper bounds. For example, as soon as we have an algorithm that solves some LCL problem in $o(\log n)$ rounds, we can immediately speed it up to $\O(\log^* n)$ rounds for free.

However, there is one key limitation: the landscape of complexities is currently understood well only as a function of the number of nodes $n$, for a constant maximum degree $\Delta = \O(1)$. For example, there are problems with complexities of the form $\Theta(\log_\Delta n)$, and problems with complexities of the form $\Theta(\log n)$, but we do not have a more fine-grained understanding of all possible complexity classes as a function of both $n$ and $\Delta$.

\paragraph{Beyond constant degrees: fundamental challenges.}

LCLs as they were originally introduced by \citet{Naor1995} only pertain to graphs of some constant maximum degree. To generalize beyond constant $\Delta$, we need a meaningful definition of the problem family.

Unfortunately, many natural generalizations lead to uninteresting results. If the local constraints may depend on the degrees in arbitrarily complicated ways, we can construct artificial problems with complexities of the form $\Theta(\log_{f(\Delta)} n)$ for virtually any function $f(\Delta) = \O(\Delta)$. For example, we could start with a problem that has complexity $\Theta(\log_\Delta n)$, and modify the problem definition so that nodes of degree $d$ will ignore up to $d-f(d)$ adjacent leaf nodes; in essence, we turn nodes of degree $d$ into nodes of degree $f(d)$, and adjust the complexity accordingly. However, this way we will learn nothing about the complexities that \emph{natural} graph problems might have.

\paragraph{Beyond constant degrees: our take.}

In this work we study the family of \emph{binary labeling problems}, as defined by \citet{balliu2019classification} (see \cref{ssec:bin-lab-prob}). These are a special case of LCLs; what is particularly attractive is that the complexity of any given binary labeling problem (in the deterministic LOCAL model) can be automatically deduced, and in the constant-degree case there are only four complexity classes in trees.

Our plan is to see exactly how the characterization of binary labeling problems can be generalized in a meaningful manner beyond constant degrees; the hope here is that this will also guide us when we seek to find the right definitions for generalizing results on all LCLs.

\subsection{Binary labeling problems}\label{ssec:bin-lab-prob}

Let us recall the key definitions from \cite{balliu2019classification}. In a binary labeling problem $\Pi$, the task is to choose a subset of edges $X \subseteq E$ in a tree $G = (V,E)$, subject to local constraints. The problem is defined as a tuple $\Pi = (d, \delta, W, B)$, where $d \in \{2,3,\dotsc\}$ is the \emph{white degree}, $\delta \in \{2,3,\dotsc\}$ is the \emph{black degree}, $W \subseteq \{0,1,\dotsc,d\}$ is the \emph{white constraint} and $B \subseteq \{0,1,\dotsc,\delta\}$ is the \emph{black constraint}. We assume that the input graph is properly $2$-colored with colors white and black. We define that $X \subseteq E$ is a solution to problem $\Pi$ if the following holds:
\begin{itemize}[noitemsep]
    \item If $v$ is a white node of degree $d$, and $v$ is incident to $k$ edges in $X$, then $k \in W$.
    \item If $v$ is a black node of degree $\delta$, and $v$ is incident to $k$ edges in $X$, then $k \in B$.
\end{itemize}
That is, we only care about the labeling incident to white nodes of degree $d$ and black nodes of degree $\delta$; these nodes are called \emph{relevant nodes}. It is often useful to imagine that white nodes represent ``nodes'' and black nodes represent ``edges'' (if $\delta = 2$) or ``hyperedges'' (if $\delta > 2$). We will usually assume that $\delta \le d$ (otherwise we can exchange the roles of black and white nodes).

\paragraph{Examples.}

Here are a couple of examples of binary labeling problems (adapted from \cite{balliu2019classification}):
\begin{enumerate}
    \item \emph{Bipartite splitting:} $W = \{1,2,\dotsc,d-1\}$ and $B = \{1,2,\dotsc,\delta-1\}$. Here the task is to split the set of edges in two classes: ``red edges'' $X$ and ``blue edges'' $E \setminus X$, and all relevant nodes must be incident to at least one red and at least one blue edge.
    \item \emph{Bipartite matching:} $W = B = \{1\}$. If we interpret that each edge $\{u,v\} \in X$ indicates that $u$ is matched with $v$, in this problem each relevant node must be matched with exactly one other node (that may or may not be relevant).
\end{enumerate}
While the problems are well-defined in any $2$-colored tree, it is often easiest to consider the case that the tree consists of only leaf nodes and relevant nodes. For example, then bipartite matching is the task of finding a matching in which all internal nodes are matched.

\paragraph{Prior work.}

Given a tuple $\Pi = (d, \delta, W, B)$, we can directly look up the round complexity of $\Pi$ in a table given by \citet{balliu2019classification} (reproduced in \cref{tab:lcl_classification}). For example, we immediately obtain:
\begin{enumerate}
    \item \emph{Bipartite splitting:} For $d = \delta = 2$ this problem requires $\Theta(n)$ rounds, but for any fixed constants $d > 2$, $\delta \ge 2$ the complexity is $\Theta_{d,\delta}(\log n)$.
    \item \emph{Bipartite matching:} For $d \ge \delta = 2$ this problem requires $\Theta_{d}(n)$ rounds, but for any fixed constants $d > 2$, $\delta > 2$ the complexity is $\Theta_{d,\delta}(\log n)$.
\end{enumerate}
However, what their classification does not capture is the complexity as a function of $d$ or $\delta$; we have made this explicit above by writing $\Theta_{d,\delta}$, to emphasize that the hidden constants may depend on $d$ and $\delta$. For example, while there are numerous problems (such as the above two examples) with a complexity $\Theta_{d,\delta}(\log n)$, it is not at all clear what is the base of logarithm in each case. In particular, which binary labeling problems get easier when $d$ and/or $\delta$ grows?

\begin{table}[tb]
    \centering
    \begin{tabular}{@{}lll@{}}
        \toprule
        \textbf{Deterministic} & \textbf{White} & \textbf{Black}\\
        \textbf{complexity} & \textbf{constraint} & \textbf{constraint}\\
        \midrule
        unsolvable & $\p\n\n^+$ & $\n\x\x^+$ \\
        & $\n\n^+\p$ & $\x\x^+\n$ \\
        & $\n\x\x^+$ & $\p\n\n^+$ \\
        & $\x\x^+\n$ & $\n\n^+\p$ \\
        
        & $\n\n\n^+$ & $\x\x\x^+$ \\
        & $\x\x\x^+$ & $\n\n\n^+$ \\
        \midrule
        $\O_{d,\delta}(1)$ & non-empty & $\p\p\p^+$ \\
        & $\p\p\p^+$ & non-empty \\
        
        & $\p\x\x^+$ & $\p\x\x^+$ \\
        & $\x\x^+\p$ & $\x\x^+\p$ \\
        \midrule
        $\Theta_{d,\delta}(n)$ & $\p\n^+\p$ & $\n\p\n$ \\
        & $\n\p\n$ & $\p\n^+\p$ \\
        
        & $\n^+\p\x$ & $\x\p\n^+$ \\
        & $\x\p\n^+$ & $\n^+\p\x$ \\
        \midrule
        $\Theta_{d,\delta}(\log n)$ & \multicolumn{2}{l}{all other cases} \\
        \bottomrule
    \end{tabular}
    \caption{The complexity landscape from \cite{balliu2019classification}.}
    \label{tab:lcl_classification}
\end{table}

\subsection{Overview of contributions and new ideas}

\paragraph{Main question.}

In this work we study parameterized families of binary labeling problems
\[
    \Pi(d,\delta) = \bigl(d, \delta, W(d), B(\delta)\bigr),
\]
and our main question is this: what can we say about the complexity of any such problem family $\Pi(d,\delta)$, as a function of $n$, $d$, and $\delta$?

\paragraph{Key technical challenge.}

We need to be careful not to make the family too broad---otherwise we will end up with an uninteresting result stating that there are artificial problems with virtually any complexity as a function of $d$ and $\delta$, without learning anything about natural graph problems.

\paragraph{Key results and new ideas.}

The main new insight is that we define the family of \emph{structurally simple} problems. This definition captures how $W(d)$ and $B(\delta)$ can depend on $d$ and~$\delta$. The aim is to exclude artificial pathological problems.

We then show that this definition is \emph{useful} in the sense that we can prove strong statements about the complexity of structurally simple problems. For example, if the complexity as a function of $n$ is $\Theta_{d,\delta}(\log n)$, then for each $d$ and $\delta$ the complexity falls in one of these fine-grained classes: $\Theta(\log_d n)$, $\Theta(\log_\delta n)$, or $\Theta(\log n)$.

Finally, we show that the definition captures a \emph{broad family} of problems: if the constraints $W(d)$ and $B(\delta)$ can be represented as binary strings in a context-free language, then $\Pi(d,\delta)$ is structurally simple. We will discuss this in more detail in \cref{ssec:intro-contrib-detail}.

To prove our main result, we also needed to develop a new, more aggressive version of the rake-and-compress technique. We discuss this in more detail in \cref{ssec:intro-rake-and-compress}.

\subsection{Contributions in more detail}\label{ssec:intro-contrib-detail}

\paragraph{Key definitions.}
In a family of binary labeling problems, both $W$ and $B$ are set families of the form $X(k) \subseteq \{0,1,\dotsc,k\}$. If we look at these set families in the examples given by \citet[Table 3]{balliu2019classification}, we observe that they all fall in one of two classes, which we call \emph{center-good} and \emph{edge-good}:

\begin{definition}[center-good]\label{def:center-good}
    A set $X \subseteq \{0,1,\dotsc,k\}$ is \emph{$\eps$-center-good} if there exists an $x \in X$ such that $k^\eps \le x \le k^{1-\eps}$.
\end{definition}

\begin{definition}[edge-good]\label{def:edge-good}
    A set $X \subseteq \{0,1,\dotsc,k\}$ is \emph{$C$-edge-good} if for all $x \in X$ we have $x \le C$ or $x \ge k-C$.
\end{definition}

Now we are ready to give the main definition:
\begin{definition}[structurally simple]\label{def:ss-problem}
    A set family $X(k) \subseteq \{0,1,\dotsc,k\}$ is \emph{structurally simple} if
    there are some constants $0 < \eps < 1$ and $C$ such that $X(k)$ is either $\eps$-center-good or $C$-edge-good for each $k$.
    A family of binary labeling problems $\Pi(d,\delta) = \bigl(d, \delta, W(d), B(\delta)\bigr)$ is \emph{structurally simple} if both $W(d)$ and $B(\delta)$ are structurally simple.
\end{definition}
For example, both the bipartite splitting problem and the bipartite matching problem are structurally simple, and so are all problems in \cite[Table 3]{balliu2019classification}; in the bipartite splitting problem, $W(d)$ and $B(\delta)$ are center-good, while in the bipartite matching problem, $W(d)$ and $B(\delta)$ are edge-good. To give a pathological example of a set family that is not structurally simple, consider, for example, $X(k) = \{ \lfloor \log k \rfloor \}$.

\paragraph{Main result.}
Our main result is a complete classification of structurally simple problems in the logarithmic region:
\begin{restatable}{theorem}{maintheorem}\label{thm:main-intro}
    Let $\Pi(d,\delta)$ be a structurally simple family of binary labeling problems
    and assume that the complexity of $\Pi(d,\delta)$ in trees is $\Theta_{d,\delta}(\log n)$.
    Then the complexity of $\Pi(d,\delta)$ in trees for each $d$ and $\delta$ falls in one of these classes,
    as long as $\delta \le d = \O(n^{1-\alpha})$ for some $\alpha > 0$:
    \begin{enumerate}[noitemsep]
        \item $\Theta(\log n)$,
        \item $\Theta(\log_\delta n)$,
        \item $\Theta(\log_d n)$.
    \end{enumerate}
\end{restatable}

\begin{remark}
A few clarifying remarks are in order that help one to interpret the result:
\begin{enumerate}
    \item Our result is constructive in the sense that we also show how to classify $\Pi(d,\delta)$ for any given $d$ and $\delta$ (see \cref{fig:log_classification} on \cpageref{fig:log_classification}).
    \item The complexity class may depend on $d$ and $\delta$. For example, we might have a problem in which for even values of $d$ the complexity is $\Theta(\log_d n)$ and for odd values of $d$ it is $\Theta(\log n)$.
    \item Throughout this work, $\Theta$-notation only hides constants that only depend on problem family $\Pi$, and not on $d$ and $\delta$; we write $\Theta_{d,\delta}$ explicitly if we are hiding constants that may depend on $d$ and $\delta$.
\end{enumerate}
\end{remark}

\paragraph{Language-theoretic justification.}

A set $X \subseteq \{0,1,\dotsc,k\}$ can be also represented as a binary string $\hat{X}$ of length $k+1$: we index the bits with $i = 0, 1, \dotsc, k$ and for each $i \in X$, the bit at index $i$ in $\hat{X}$ is $1$, and otherwise $0$. Given a set family $X(k)$, we can then define a language of binary strings
\[
    \hat{X} = \bigl\{ \hat{X}(2), \hat{X}(3), \hat{X}(4), \dotsc \bigr\}.
\]
Note that $\hat{X}$ is by construction \emph{thin}: it contains at most one word of any given length \cite{paun1995thin}.

\begin{example}
If $X(k) = \{1,2,\dotsc,k-1\}$, then $\hat{X}(k) = 01^{k-1}0$ and
$
    \hat{X} = \{ 010, 0110, 01110, \dotsc \} = 01^{+}0
$.
If $X(k) = \{1\}$, then $\hat{X}(k) = 010^{k-1}$ and
$
    \hat{X} = \{ 010, 0100, 01000, \dotsc \} = 010^{+}
$.
Here we use $1^\ell$ to denote a sequence of $\ell$ $1$s, and $1^{+}$ to denote the sequence of one or more $1$s.
\end{example}

Equipped with this notation, we can study families of binary labeling problems from a language-theoretic perspective: any given problem family $\Pi(d,\delta)$ can be specified by giving a pair of languages $(\hat{W},\hat{B})$. As long as these are languages over the binary alphabet, and there is exactly one string of each length $2, 3, \dotsc$, such a pair of languages can be interpreted as a parameterized problem family $\Pi(d,\delta)$.

Many example problems in prior work \cite{balliu2019classification} correspond to \emph{regular languages}. For example, bipartite splitting is $(01^+0, 01^+0)$, and bipartite matching is $(010^+, 010^+)$, while the \emph{sinkless orientation problem} \cite{brandt16lll} is $(11^+0, 011^+)$. We take one step beyond regular languages and consider the case of \emph{context-free languages}. The key observation is summarized in the following lemma:
\begin{restatable}{lemma}{sscontextfreelemma}\label{lem:ss-context-free}
    Let $X(k) \subseteq \{0,1,\dotsc,k\}$ be a family of sets. If $\hat{X}$ is a context-free language, then $X(k)$ is structurally simple.
\end{restatable}
\begin{corollary}\label{cor:ss-context-free}
    Let $\Pi(d,\delta) = (d, \delta, W(d), B(\delta))$ be a family of binary labeling problems.
    If $\hat{W}$ and $\hat{B}$ are context-free languages,
    then $\Pi(d,\delta)$ is structurally simple.
\end{corollary}

This is the main justification for focusing on structurally simple problems: we have only excluded some pathological cases that are so complicated that they cannot be expressed with context-free grammars (this is also the reason why we call them structurally simple).

\subsection{Key building block: a new rake-and-compress variant}\label{ssec:intro-rake-and-compress}

A key tool for designing $\O(\log n)$-time algorithms for binary labeling problems in trees has been the \emph{rake-and-compress technique} \cite{Miller1985}. A typical application of this technique proceeds as follows; we alternate between two steps:
\begin{enumerate}[noitemsep]
    \item Eliminate leaf nodes and isolated nodes.
    \item Eliminate sufficiently long paths.
\end{enumerate}
If we apply these steps for $\O(\log n)$ times, in a tree of any shape, we will eliminate all nodes. Then we can construct a solution by working backwards: repeatedly put back one layer of nodes and construct a solution that makes these nodes happy; see \cite{balliu2019classification} for more detailed examples.

From our perspective, the main drawback of the rake-and-compress technique is that it does not benefit from high-degree nodes. For example, if we apply it in a tree in which all internal nodes have degree $s$, the worst-case complexity is still $\Theta(\log n)$, not $\Theta(\log_s n)$.

We introduce a new version of the technique that strictly benefits from high-degree nodes. Let us first rephrase the usual rake-and-compress procedure as follows (note that nodes in the middle of long paths are also nodes that are not close to higher-degree nodes):
\begin{enumerate}[noitemsep]
    \item Eliminate leaf nodes and isolated nodes.
    \item Eliminate nodes that are not close to any node of degree $3$ or more.
\end{enumerate}
We generalize this as follows, so that we are much more aggressive with the compression step:
\begin{enumerate}[noitemsep]
    \item Eliminate leaf nodes and isolated nodes.
    \item Eliminate nodes that are not close to any node of degree $s$ or more.
\end{enumerate}
We show that this leads to a procedure that completes in $\O(\log_s n)$ rounds, and we show that we can use this more aggressive version of rake-and-compress to solve many binary labeling problems. This is the key building block that leads to the upper bounds $\O(\log_d n)$ and $\O(\log_\delta n)$ in \cref{thm:main-intro}.

We believe that our new rake-and-compress technique will find applications also beyond binary labeling problems.

\subsection{Questions for future work}

While our classification is constructive, it also gives rise to a number of new open questions related to the \emph{decidability} of complexities for entire problem families; we present here one example:
\begin{question}
    Given context-free grammars for thin languages $\hat{W}$ and $\hat{B}$ that define a problem family $\Pi(d,\delta)$, how hard is it to decide if the complexity of $\Pi(d,\delta)$ is $\Theta(\log_d n)$ for all but finitely many values of $d$ and $\delta$?
\end{question}

\subsection{Roadmap}

We first classify structurally simple problems in the logarithmic region:
\begin{itemize}
    \item In \cref{sec:log-lower} we prove lower bounds of the forms $\Omega(\log n)$,  $\Omega(\log_d n)$, and  $\Omega(\log_\delta n)$.
    \item In \cref{sec:log-upper} we prove upper bounds of the forms $\O(\log n)$,  $\O(\log_d n)$, and  $\O(\log_\delta n)$. Here we also introduce and use our new rake-and-compress technique.
    \item In \cref{sec:log-classification} we put together the lower and upper bounds and establish a full classification in the logarithmic region. This proves the main result, \cref{thm:main-intro}.
\end{itemize}
Then in \cref{sec:ss-languages} we show that all problems that can be defined with context-free grammars are indeed structurally simple; this establishes \cref{lem:ss-context-free,cor:ss-context-free}.

In \cref{sec:other-classification} we classify problems in the constant and linear complexity classes; this leads to a complete classification that we summarize in \cref{tab:results}.

\begin{table}[tb]
    \centering
    \begin{tabular}{@{}lll@{}}
        \toprule
        \textbf{Broad class}
        & \textbf{Fine-grained classes}
        & \textbf{Reference}\\
        \midrule
        $\O_{d,\delta}(1)$
        & $\O(1)$
        & \cref{prop:constant} \\
        \midrule
        $\Theta_{d,\delta}(\log n)$
        & $\Theta(\log_d n)$
        & \cref{thm:main-intro} \\
        & $\Theta(\log_\delta n)$ \\
        & $\Theta(\log n)$ \\
        \midrule
        $\Theta_{d,\delta}(n)$
        & $\Theta(n/(d+\delta))$
        & \cref{cor:linear} \\
        \midrule
        unsolvable
        & unsolvable \\
        \bottomrule
    \end{tabular}
    \caption{Summary of our results for structurally simple problems (for precise technical assumptions please refer to the theorem statements).}\label{tab:results}
\end{table}

\subsection{Model}\label{ssec:prelim}

Let $G = (V, E)$ be a graph with $n$ nodes.
We work in the \emph{deterministic LOCAL model} \cite{Linial1992,Peleg2000}: Each node $v \in V$ is assigned a \emph{unique identifier} $\operatorname{id}(v) \in \{1, 2, \ldots, n^c\}$ for some constant $c$.
Initially each node knows its own identifier, degree, the total number of nodes $n$, and its input label (here: its color, black or white).
All nodes execute the same algorithm, and computation proceeds in synchronized rounds. In each round, nodes transmit messages of arbitrary size to their neighbors, receive messages, and perform local deterministic computations of arbitrary complexity.
Eventually each node must stop and produce its local output (here: which of its incident edges are in the solution $X \subseteq E$).
The running time, or complexity, of the algorithm is defined as the number of rounds required by all nodes to make local output decisions. Note that an algorithm that runs in $T$ rounds can also be interpreted as a mapping from the radius-$T$ neighborhood of each node to its local output.

\section{Logarithmic lower bounds}\label{sec:log-lower}

We start by proving the main lower bound results. We establish the lower bounds of $\Omega(\log_d n)$, $\Omega(\log_\delta n)$, and $\Omega(\log n)$, depending on the structure of the problem.

\subsection{A general lower bound}

We start by showing a general lower bound that holds for any problem:
\begin{lemma}\label{lem:log_lower_bound}
    Let $\Pi(d, \delta)$ be a family of binary labeling problems.
    For each $d$ and $\delta$ the complexity of $\Pi(d,\delta)$ is either $\O(1)$ or $\Omega(\log_d n)$, assuming $\delta \le d$.
\end{lemma}
\begin{proof}
    As proved by \citet{balliu2019classification}, all binary labeling problems that are not solvable with locality~$\O(1)$ have complexity $\Omega_{d,\delta}(\log n)$.
    They prove this by reducing all such problems to the \emph{forbidden degree or sinkless orientation}, which is proved to be a fixed point of round elimination and not $0$-round solvable.

    It is known that there exists bipartite graphs with high girth.
    In particular, for any $a, b \in \N$, there exists $(a,b)$-biregular graphs with girth $\Theta(\log_{ab} n)$ \cite{furedi1995biregular}.
    By using techniques similar to \citet{balliu2023sinkless} on these graphs, we get an $\Omega(\log_d n)$ lower bound for $\Pi$.
\end{proof}

\subsection{Definitions and problem transformations}

To get more fine-grained lower bounds, we first define a few transformations for problems.
With the help of these transformations, we can apply \cref{lem:log_lower_bound} to prove stronger lower bounds.

We start by introducing the switch and the reverse of a problem and observe that they do not affect the complexity of the problem:
\begin{definition}[switch of a problem]
    Let $\Pi(d, \delta) = (d, \delta, W(d), B(\delta))$ be a family of problems.
    The \emph{switch} of $\Pi(d, \delta)$ is $\Pi^s(\delta, d) = (\delta, d, B(\delta), W(d))$.
\end{definition}
\begin{definition}[reverse of a set family]
    Let $X(k) \subseteq \{0, \dots, k\}$.
    The \emph{reverse} of $X(k)$ is $X^r(k) = \{k - x \mid x \in X\}$.
\end{definition}
\begin{definition}[reverse of a problem]
    Let $\Pi(d, \delta) = (d, \delta, W(d), B(\delta))$ be a family of problems.
    The \emph{reverse} of $\Pi(d, \delta)$ is $\Pi^r(d, \delta) = (d, \delta, W^r(d), B^r(\delta))$.
\end{definition}

\begin{lemma}\label{lem:switch}
    Binary labeling problem $\Pi^s(\delta, d)$ has the same complexity as $\Pi(d, \delta)$.
\end{lemma}
\begin{proof}
    Given an algorithm solving $\Pi$, we get an algorithm solving $\Pi^s$ by reversing the role of white and black nodes.
\end{proof}
\begin{lemma}\label{lem:reverse}
    Binary labeling problem $\Pi^r(d, \delta)$ has the same complexity as $\Pi(d, \delta)$.
\end{lemma}
\begin{proof}
    Given an algorithm solving $\Pi$, we get an algorithm solving $\Pi^r$ by replacing the solution output with its complement.
\end{proof}

We now define a more complicated transformation called \emph{shift} of a problem; see \cref{tab:shift-example} for an example:
\begin{definition}[shift of a set family]
    Let $X(d) \subseteq \{0, 1, \dots, d\}$.
    A \emph{shift} by $k$ of $X(d)$ is
    \begin{equation*}
        X^{\leftarrow k}(d-k) = \{0, 1, \dots, d-k\} \cap \bigl\{x - i \mid x \in X(d), i \in \{0, 1, \dots, k\} \bigr\} .
    \end{equation*}
\end{definition}

\begin{table}[tb]
    \centering
    \begin{tabular}{l c ccccc ccccc ccccc ccccc}
        \toprule
        $X(20)$:
            &$\cdot$
            &$\cdot$&$\cdot$&$\cdot$&$\cdot$&5
            &$\cdot$&$\cdot$&$\cdot$&$\cdot$&$\cdot$
            &$\cdot$&$\cdot$&$\cdot$&$\cdot$&15
            &$\cdot$&$\cdot$&$\cdot$&$\cdot$&$\cdot$ \\
        \midrule
        $X^{\leftarrow 1}(19)$:
            &$\cdot$
            &$\cdot$&$\cdot$&$\cdot$&4&5
            &$\cdot$&$\cdot$&$\cdot$&$\cdot$&$\cdot$
            &$\cdot$&$\cdot$&$\cdot$&14&15
            &$\cdot$&$\cdot$&$\cdot$&$\cdot$ \\
        $X^{\leftarrow 2}(18)$:
            &$\cdot$
            &$\cdot$&$\cdot$&3&4&5
            &$\cdot$&$\cdot$&$\cdot$&$\cdot$&$\cdot$
            &$\cdot$&$\cdot$&13&14&15
            &$\cdot$&$\cdot$&$\cdot$ \\
        $X^{\leftarrow 3}(17)$:
            &$\cdot$
            &$\cdot$&2&3&4&5
            &$\cdot$&$\cdot$&$\cdot$&$\cdot$&$\cdot$
            &$\cdot$&12&13&14&15
            &$\cdot$&$\cdot$ \\
        $X^{\leftarrow 4}(16)$:
            &$\cdot$
            &1&2&3&4&5
            &$\cdot$&$\cdot$&$\cdot$&$\cdot$&$\cdot$
            &11&12&13&14&15
            &$\cdot$ \\
        $X^{\leftarrow 5}(15)$:
            &0
            &1&2&3&4&5
            &$\cdot$&$\cdot$&$\cdot$&$\cdot$&10
            &11&12&13&14&15 \\
        $X^{\leftarrow 6}(14)$:
            &0
            &1&2&3&4&5
            &$\cdot$&$\cdot$&$\cdot$&9&10
            &11&12&13&14 \\
        $X^{\leftarrow 7}(13)$:
            &0
            &1&2&3&4&5
            &$\cdot$&$\cdot$&8&9&10
            &11&12&13 \\
        \midrule
        $X^{\leftarrow 10}(10)$:
            &0
            &1&2&3&4&5
            &6&7&8&9&10 \\
        $X^{\leftarrow 15}(5)$:
            &0
            &1&2&3&4&5 \\
        \bottomrule
    \end{tabular}
    \caption{Examples of the definition of a shift, assuming $X(20) = \{ 5,15 \}$.}\label{tab:shift-example}
\end{table}

\begin{definition}[shift of a problem]
    Let $\Pi(d, \delta) = (d, \delta, W(d), B(\delta))$ be a family of problems.
    The \emph{white shift} of $\Pi$ is \[\Pi^{\leftarrow_W k}(d, \delta) = (d, \delta, W^{\leftarrow k}(d), B(\delta)),\]
    and the \emph{black shift} of $\Pi$ is \[\Pi^{\leftarrow_B k}(d, \delta) = (d, \delta, W(d), B^{\leftarrow k}(\delta)).\]
\end{definition}

\begin{lemma}\label{lem:shift}
    Let $\Pi(d, \delta)$ be a family of problems with complexity $f(n, d, \delta)$.
    Then the white shift $\Pi^{\leftarrow_W k}$ has complexity $\O(f((k+1)n,d+k,\delta))$ and the black shift $\Pi^{\leftarrow_B k}$ has complexity $\O(f((k+1)n,d,\delta+k))$
\end{lemma}
\begin{proof}
    Let $\Pi(d, \delta) = (d, \delta, W(d), B(\delta))$ be family of binary labeling problems.
    We show that the white shift $\Pi^{\leftarrow_W k}$ has complexity $\O(f((k+1)n,d+k,\delta))$; the proof for the black shift $\Pi^{\leftarrow_B k}$ follows by \cref{lem:switch}.

    Fix parameters $d$, $\delta$ and $k$.
    Let tree $G$ be an instance for problem $\Pi^{\leftarrow_W k}$ with $n$ nodes.
    We are only interested in black nodes of degree $\delta$ and white nodes of degree $d$ as the rest of the nodes are unrestricted.
    Let $G'$ be a copy of $G$ with $k$ additional black leaves attached to each white node of degree $d$.
    Note that $G'$ has at most $n'=(k+1)n$ nodes.
    Let $v$ be a white node of $G$, and let $v'$ be the corresponding white node in $G'$.
    Denote the set of edges incident to $v$ by $E_{G}(v)$, and edges incident to $v'$ by $E_{G'}(v')$.

    Suppose that $X$ is a solution for $\Pi(d+k,\delta)$ on $G'$; such a solution can be computed from input graph $G$ with locality $\O(f((k+1)n, d+k, \delta))$ by assumption.
    Then $X$ induces a solution $Y = X \cap E(G)$ for $\Pi^{\leftarrow_W k}$ on $G$.
    We know that $\deg_X(v') = \bigl|E_{G'}(v') \cap X\bigr|\in W(d+k)$.
    We now show that $\deg_Y(v) = \bigl|E_G(v)\cap Y\bigr|\in W^{\leftarrow k}(d)$.

    It is clear that $0 \le \deg_Y(v) \le (d + k) - k = d$.
    Moreover,
    \begin{equation*}
        \begin{aligned}
            \deg_Y(v) = \deg_X(v') - &\bigl|(E_{G'}(v') \setminus E_G(v))\cap X\bigr|, \quad \text{and} \\
            0 \leq &\bigl|(E_{G'}(v') \setminus E_G(v))\cap X\bigr| \leq \bigl|E_{G'}(v') \setminus E_G(v)\bigr| = k .
        \end{aligned}
    \end{equation*}
    Therefore, $\deg_X(v') - k \le \deg_Y(v) \le \deg_X(v')$.
    Combining these gives
    \begin{equation*}
        \deg_Y(v) \in \{0, 1, \dots, d\} \cap \bigl\{\deg_X(v') - k, \dots, \deg_X(v')\bigr\} \subseteq W^{\leftarrow k}(d),
    \end{equation*}
    completing the proof.
\end{proof}

With the help of these transformations, we are now ready to prove lower bounds for problems.
Indeed, if we can show that some problem $\Pi^{\leftarrow_W k}(d, \delta)$ has complexity $\Omega(f(n,d,\delta))$, then we also know that $\Pi$ has complexity $\Omega(f(n/(k+1),d-k,\delta))$.

\subsection{Problem reductions}
In this section, we use \cref{lem:switch,lem:reverse,lem:shift} to reduce problems to easier ones and give lower bounds on those with the help of \cref{lem:log_lower_bound}. In what follows, we slightly abuse the notation and represent the sets $W(d)$ and $B(\delta)$ using the corresponding binary strings $\hat{W}(d)$ and $\hat{B}(\delta)$.

\begin{proposition}\label{prop:orientation_lb}
    Let $k$ and $l$ be integers.
    The problem $\Pi(d, \delta) = (d, \delta, \n^{d-k+1} \p^k, \p^l \n^{\delta-l+1})$ for $d \ge k, \delta \ge l$ is either unsolvable or has complexity $\Omega(\log n)$ when $d, \delta = \O(n^{1-\eps})$ for some $\eps > 0$.
\end{proposition}
\begin{proof}
    Let $W = \n^{d-k+1} \p^k$ and $B = \p^l \n^{\delta-l+1}$.
    We reduce $\Pi$ to $\Pi'(k, l) = (k, l, \n\p^k, \p^l\n)$ by first applying white shift to get $\Pi_1= \Pi^{\leftarrow_W d-k}$ and then black shift to get $\Pi' = \Pi_2 = \Pi_1^{\leftarrow_B \delta-l}$.
    
    By \cref{lem:log_lower_bound}, $\Pi'(k, l)$ has complexity $\Omega(\log n)$.
    Indeed, $k$ and $l$ are constants, so one can verify this by consulting the classification table of \citet{balliu2019classification}.
    Therefore, by \cref{lem:shift}, $\Pi_1$ has complexity $\Omega(\log(n/(\delta-l+1)))$.
    Applying \cref{lem:shift} again, we get that $\Pi$ has complexity
    \begin{equation*}
        \Omega\left(\log\left(\frac{n}{(\delta-l+1)(d-k+1)}\right)\right)=\Omega(\log n)
    \end{equation*}
    when $d, \delta = \O(n^{1-\eps})$.
\end{proof}
\begin{remark}
    Notice that if $k \leq 2$ and $l\leq 2$ then problem $\Pi(d, \delta) = (d, \delta, \n^{d-k+1} \p^k, \p^l \n^{\delta-l+1})$ is either unsolvable or has complexity $\Theta(n)$.
\end{remark}

\begin{proposition}\label{prop:w_sparse_lb}
    Let $k$ and $l$ be integers.
    The problem $\Pi(d, \delta) = (d, \delta, \p^k \n^{d-k-l+1} \p^l, \n \p^{\delta-1} \n)$ has complexity $\Omega(\log_\delta n)$ when $d = \O(n^{1-\eps})$ for some $\eps > 0$.
\end{proposition}
\begin{proof}
    We apply white shift to reduce $\Pi(d, \delta)$ to problem \[\Pi'(k+l, \delta) = \Pi^{\leftarrow_W d-k-l} = (k+l, \delta, \p^k \n \p^l,\n\p^{\delta-1} \n).\]
    Again, by \cref{lem:log_lower_bound}, problem $\Pi'$ has complexity $\Omega(\log_\delta n)$.
    By previous results~\cite{balliu2019classification}, we know that $\Pi'$ is neither a trivial nor an unsolvable problem.
    Applying \cref{lem:shift}, we get that $\Pi$ has complexity $\Omega(\log_\delta(n/(d-k-l+1)))=\Omega(\log_\delta n)$ when $d = \O(n^{1-\eps})$.
\end{proof}

\section{Logarithmic upper bounds}\label{sec:log-upper}

Now we proceed to prove upper bounds; later in \cref{sec:log-classification} we will see that our lower and upper bounds indeed provide tight bounds for all structurally simple problems in the logarithmic region.

\subsection{Center-good problems}

Recall the concept of center-good constraints that we introduced in \cref{def:center-good}. We start by showing that problems whose white or black constraint is center-good can be solved efficiently:
\begin{lemma}\label{lem:center-good-algo}
    Let $\Pi(d, \delta) = (d, \delta, W(d), B(\delta))$ be a solvable structurally simple problem.
    If $W(d)$ is center-good, then $\Pi(d, \delta)$ can be solved with locality $\O(\log_d n)$.
    Similarly, if $B(\delta)$ is center-good, then $\Pi(d, \delta)$ can be solved with locality $\O(\log_\delta n)$.
\end{lemma}

To prove the lemma, we will first show that problems of this type are \emph{$(s,t)$-resilient}~\cite{balliu2019classification} for some $s, t \in \N$:
\begin{definition}[resilient problem~\cite{balliu2019classification}]
    A binary labeling problem $\Pi(d, \delta) = (d, \delta, W(d), B(\delta))$ is \emph{$(s,t)$-resilient} if
    \begin{itemize}[noitemsep]
        \item string $\hat W(d)$ does not contain a substring of the form $\n^{d+1-s}$ and
        \item string $\hat B(\delta)$ does not contain a substring of the form $\n^{\delta+1-t}$.
    \end{itemize}
\end{definition}

Resilient problems are useful because they allow partial labelings to be completed:
\begin{definition}[partial labeling]
    Let $G = (V, E)$ be a tree.
    We call $\ell \colon E' \to \{\n, \p\}$ for some subset of edges $E' \subseteq E$ a \emph{partial labeling of $G$}.

    Given two partial labelings $\ell_1 : E_1 \to \{\n, \p\}$ and $\ell_2 : E_2 \to \{\n, \p\}$ on $G$, we say that $\ell_2$ is a \emph{completion} of $\ell_1$ if $E_1\subseteq E_2$, that is, it is defined on a larger set of edges.
    
    Given a problem $\Pi$ and a partial labeling $\ell$, we say a node $v \in V$ is \emph{labelled in a valid manner} if the set of edges incident to $v$ is in the domain of $\ell$ and if $\ell$ satisfies the constraints of $\Pi$ on $v$.
\end{definition}

\begin{lemma}[\citet{balliu2020classification-arxiv}]\label{prop:resiliency}
    Let $\Pi(d, \delta)$ be a $(t, s)$-resilient problem, let $G = (V, E)$ be a tree, and let $\ell$ be a partial labeling on $G$.
    There for every node $v \in V$, there exists a completion of $\ell$ that labels $v$ in a valid manner if one of the following conditions holds:
    \begin{itemize}[noitemsep]
        \item $v$ is a white node incident to at most $t$ edges with non-empty labels, or
        \item $v$ is a black node incident to at most $s$ edges with non-empty labels.
    \end{itemize}
\end{lemma}

Our plan is to recursively remove layers of nodes from the input tree. Then we can put the layers back and complete the layering by exploiting \cref{prop:resiliency}. We first define a procedure $\DEG$:
\begin{definition}[$\DEG$]\label{def:deg}
    Let $G = (V_W \sqcup V_B, E)$ be a tree with bipartition $(V_W,V_B)$.
    Let
    \begin{equation*}
        \degen_{s,t}(G) = \{v \in V_W \mid \deg_G(v) \le s \} \cup \{v \in V_B \mid \deg_G(v) \le t \} .
    \end{equation*}
    Procedure $\DEG(s,t)$ partitions nodes of $G$ into non-empty sets $L_1, L_2, \dots, L_k = \emptyset$ for some $k$ as follows:
    \begin{align*}
        G_0 & = G, \\
        L_{i+1} & = \degen_{s,t}(G_i) \quad \text{if } G_i\text{ is not empty}, \\
        G_{i+1} & = G_i \setminus L_{i+1}.
    \end{align*}
\end{definition}

We now show that computing decomposition $\DEG(s, t)$ can be computed with locality $\O(\log_{\max\{s,t\}} n)$ in the LOCAL model.
We start with the following lemma showing that the number of high-degree nodes in a tree is bounded:
\begin{lemma}\label{lem:bipartite_degs}
    Let $d \ge 1$ and let $G(V_W \sqcup V_B, E)$ be a tree with bipartition $(V_W, V_B)$, where $V_W$ is the set of white nodes and $V_B$ is the set of black nodes.
    There are at most $(|V_B|-1)/(d-1)$ white nodes with degree at least $d$, and at most $(|V_W|-1)/(d-1)$ black nodes with degree at least $d$.
\end{lemma}
\begin{proof}
    Let $d$ and $G$ be as in the statement.
    We show the statement for white nodes; the proof for black nodes follows by symmetry.

    Denote by $n^W_{\ge d}$ the number of white nodes with degree at least $d$, and by $n^W_{< d}$ the number of the rest of white nodes.
    By the handshake lemma for bipartite graphs, we have
    \begin{equation}\label{eq:bipartite_degs_edges}
        |E| \ge d n^W_{\ge d} + n^W_{< d} = (d - 1) n^W_{\ge d} + |V_W|.
    \end{equation}
    As $G$ is a tree, we have $|E| = |V_W| + |V_B| - 1$.
    Combining this with \eqref{eq:bipartite_degs_edges} gives the result
    \begin{equation*}
        n^W_{\ge d} \le (|V_B| - 1) / (d - 1) . \qedhere
    \end{equation*} 
\end{proof}
We can now prove the following lemma showing that each round of $\DEG$ reduces the number of nodes by a constant factor:
\begin{lemma}
    Let $s, t \ge 1$, and let $d \ge 1$ and let $G(V_W \sqcup V_B, E)$ be a tree with bipartition $(V_W, V_B)$.
    Running two iterations of $\DEG(s, t)$ on $G$ reduces the number of nodes in $G$ by a factor of $\Omega(st)$.
\end{lemma}
\begin{proof}
    Let $s, t$ and $G$ be as in the statement.
    We call nodes of $V_W$ white nodes and nodes of $V_B$ black nodes.
    Let $L_1 = \degen_{s,t}(G)$, $G_1 = G \setminus L_1$,  $L_2 = \degen_{s,t}(G_1)$, and $G_2 = G_1 \setminus L_2$.
    Let $n^W_1$ and $n^B_1$ be the number of white and black nodes in $L_1$, respectively, and let $n^W_2$ and $n^B_2$ be the number of white and black nodes in $L_2$.

    By \cref{lem:bipartite_degs}, we have
    \[
        n^W_1 \le (|V_B|-1)/(t-1) \text{ and } n^B_1 \le (|V_W|-1)/(s-1).
    \]
    Applying \cref{lem:bipartite_degs} again on $G_1$, we have
    \[
        n^W_2 \le (n^B_1-1)/(t-1) \text{ and } ^B_2 \le (n^W_1-1)/(s-1) .
    \]
    Expanding the definitions of $n^W_1$ and $n^B_1$ gives the result:
    \[
        n^W_2 \le |V_W| / \Omega(st) \text{ and } n^B_2 \le |V_B| / \Omega(st) . \qedhere
    \]
\end{proof}
Combining this lemma with the simple fact that all nodes can locally compute $\degen_{s,t}$ for each round, we immediately get the locality of $\DEG(s, t)$:
\begin{corollary}\label{cor:deg_running_time}
   The running time of procedure $\DEG(s,t)$ is $\O(\log_{st} n) = \O(\log_{\max\{s,t\}} n)$ in the LOCAL model.
\end{corollary}

We can now show that resilient problems can be solved efficiently in the LOCAL model:
\begin{lemma}\label{lem:resilient_algo}
    Any $(s,t)$-resilient problem $\Pi(d, \delta)$ can be solved with locality $\O(\log_{\max\{s,t\}} n)$.
\end{lemma}
\begin{proof}
    Let $G$ be the input tree.
    Start by computing layer decomposition $L_1, \dots, L_k$ for $G$; this can be done with locality $\O(\log_{\max\{s,t\}} n)$ by \cref{cor:deg_running_time}.
    Label each edge that is between two nodes in the same layer with \n.

    We can now proceed to label the rest of the edges layer-by-layer, from $k-1$ to $1$.
    The nodes on layer $i$ label all their adjacent edges in a valid manner.
    This is possible by \cref{prop:resiliency} as white (respectively black) nodes in layer $i$ have at most $s$ (respectively $t$) edges  to layer $i$ or higher layers, and only those edges have a label.
    Note that this procedure can be done in parallel by all nodes of layer $i$ as all edges with both endpoints in the same layer have their output fixed at \n.
    Finally, by definition of $L_1$, white nodes on the layer $1$ have at most $s < d$ neighbors, and black nodes on the layer $1$ have at most $t < \delta$ neighbors, hence they do not impose any restrictions on their adjacent edges.
\end{proof}

Now we are finally ready to prove \cref{lem:center-good-algo}:

\begin{proof}[Proof of \cref{lem:center-good-algo}]
    Let us focus on the case where $W(d)$ is center-good; the case where $B(\delta)$ is center-good is symmetric.
    By definition, any sufficiently large set $w \in W(d)$ contains an element $p \in [d^\eps,d^{1-\eps}]$.
    Therefore, the problem is $(d^{\min\{\eps, 1-\eps\}},1)$-resilient, and the rest follows by \cref{lem:resilient_algo}.
\end{proof}

\subsection{Aggressive rake-and-compress}
In this section, we introduce the aggressive rake-and-compress algorithm, compute its complexity, and give some of its key properties.
\begin{definition}
    For a graph $G$ and some positive integer $r$, we define the following two sets:
    \begin{align*}
        \leaves(G) &= \{v \in V(G) \mid \deg_G(v) = 1\}, \\
        \ext_{r,\Delta}(G) &= \{v \in V(G) \mid \forall u\in N^r[v], \deg_G(u) < \Delta\} = V(G)\setminus N^r[\{v \in V(G)\mid \deg_G(v) \geq \Delta \}],
    \end{align*}
    where $N^r[v]$ in the distance-$r$ closed neighborhood of $v$.
    Removing the first set from the graph is a \emph{rake} operation, removing the second set from the graph is a \emph{compress} operation.
\end{definition}

By computing these sets recursively and removing them from the graph, we can partition the graph into layers. This is a very helpful tool to create efficient $\O(\log_\Delta n)$ local algorithms.

\begin{definition}
    Procedure $\ARC(r,\Delta)$ partitions the set of nodes $V$ into non-empty sets $L_1, L_2, \dots, L_k$ for some $L$ as follows:
    \begin{align*}
        G_0     & = G,                                                                           \\
        L_{i+1} & = \leaves(G_i) \cup \ext_{r,\Delta}(G_i) \quad \text{ if } G_i \text{ is not empty},\\
        G_{i+1} & = G_i \setminus L_{i+1}.
    \end{align*}
\end{definition}

In the following, we show that computing $\ARC(r,\Delta)$ requires only $\O(r^2\cdot \log_{\Delta/r} n)$ rounds.

\begin{lemma}\label{lem:mrc_shrinkage}
    Let $G=(V,E)$ be a tree and $r\in \{1,2,\dotsc\}$. Apply successively the distance-$r$ rake operation, compress operation and then rake operation again, $r$ times successively. If $V'$ is the remaining vertex set then $|V'| \leq \O(r)|V|/\Omega(\Delta)$.
\end{lemma}
\begin{proof}
    Let $G$ be a tree.
    We color the nodes of the tree as follows:
    \begin{itemize}[noitemsep]
        \item a node of degree $\Delta$ is colored blue,
        \item a non-leaf node of degree $< \Delta$ at distance $\leq r$ from a node of degree $\Delta$ is colored red,
        \item and all other nodes are colored white.
    \end{itemize}
    Notice that after a rake operation and a compress operation, all white nodes are removed.
    To account for the remaining rake operation, we recolor all red nodes as follows; see \cref{fig:recoloring}:
    \begin{itemize}[noitemsep]
        \item an old red node is recolored red if it is adjacent to at least two non-white nodes,
        \item and all the other old red nodes are recolored white.
    \end{itemize}
    We do this recoloring $r$ times.

    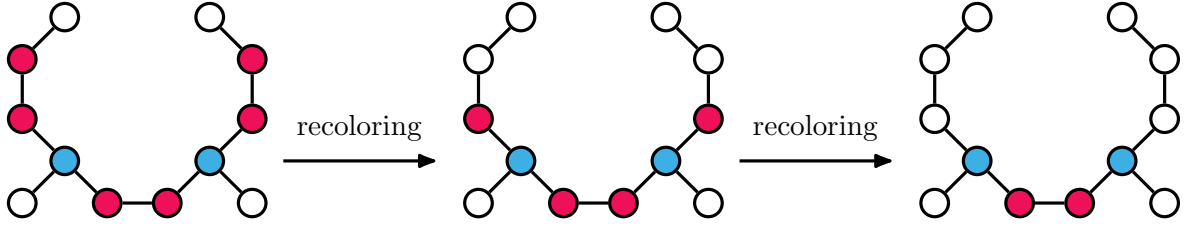
\begin{figure}
        \centering
        \begin{tikzpicture}[every node/.style={draw=black,circle, very thick},node distance=1em,scale=0.5]
        \begin{scope}                
            \node[fill=CornflowerBlue] (blue1) {};
            \node[fill=white, below left=of blue1] (white1) {};
            \node[fill=OrangeRed, below right=of blue1] (red1) {};
            \node[fill=OrangeRed, right=of red1] (red2) {};
            \node[fill=CornflowerBlue, above right=of red2] (blue2) {};
            \node[fill=white, below right=of blue2] (white2) {};
            \node[fill=OrangeRed, above left=of blue1] (red3) {};
            \node[fill=OrangeRed, above=of red3] (red4) {};
            \node[fill=white, above right=of red4] (white3) {};
            \node[fill=OrangeRed, above right=of blue2] (red5) {};
            \node[fill=OrangeRed, above=of red5] (red6) {};
            \node[fill=white, above left=of red6] (white4) {};  
    
            \draw[very thick] (blue1) -- (white1);
            \draw[very thick] (blue1) -- (red1);
            \draw[very thick] (red2) -- (red1);
            \draw[very thick] (blue1) -- (red3);
            \draw[very thick] (red4) -- (red3);
            \draw[very thick] (red4) -- (white3);
            \draw[very thick] (blue2) -- (red2);
            \draw[very thick] (blue2) -- (white2);
            \draw[very thick] (blue2) -- (red5);
            \draw[very thick] (red6) -- (red5);
            \draw[very thick] (red6) -- (white4);
        \end{scope}
    
            \draw[-{Latex[round]}, very thick] (5.75,0) -- (9.75,0) node[draw=none,fill=none,midway,above=-1.5em] {recoloring};       
        \begin{scope}[shift={(12,0)}]
            \node[fill=CornflowerBlue] (blue1) {};
            \node[fill=white, below left=of blue1] (white1) {};
            \node[fill=OrangeRed, below right=of blue1] (red1) {};
            \node[fill=OrangeRed, right=of red1] (red2) {};
            \node[fill=CornflowerBlue, above right=of red2] (blue2) {};
            \node[fill=white, below right=of blue2] (white2) {};
            \node[fill=OrangeRed, above left=of blue1] (red3) {};
            \node[fill=white, above=of red3] (red4) {};
            \node[fill=white, above right=of red4] (white3) {};
            \node[fill=OrangeRed, above right=of blue2] (red5) {};
            \node[fill=white, above=of red5] (red6) {};
            \node[fill=white, above left=of red6] (white4) {};  
    
            \draw[very thick] (blue1) -- (white1);
            \draw[very thick] (blue1) -- (red1);
            \draw[very thick] (red2) -- (red1);
            \draw[very thick] (blue1) -- (red3);
            \draw[very thick] (red4) -- (red3);
            \draw[very thick] (red4) -- (white3);
            \draw[very thick] (blue2) -- (red2);
            \draw[very thick] (blue2) -- (white2);
            \draw[very thick] (blue2) -- (red5);
            \draw[very thick] (red6) -- (red5);
            \draw[very thick] (red6) -- (white4);
        \end{scope}
        \draw[-{Latex[round]}, very thick] (17.75,0) -- (21.75,0) node[draw=none,fill=none,midway,above=-1.5em] {recoloring};       
        \begin{scope}[shift={(24,0)}]
            \node[fill=CornflowerBlue] (blue1) {};
            \node[fill=white, below left=of blue1] (white1) {};
            \node[fill=OrangeRed, below right=of blue1] (red1) {};
            \node[fill=OrangeRed, right=of red1] (red2) {};
            \node[fill=CornflowerBlue, above right=of red2] (blue2) {};
            \node[fill=white, below right=of blue2] (white2) {};
            \node[fill=white, above left=of blue1] (red3) {};
            \node[fill=white, above=of red3] (red4) {};
            \node[fill=white, above right=of red4] (white3) {};
            \node[fill=white, above right=of blue2] (red5) {};
            \node[fill=white, above=of red5] (red6) {};
            \node[fill=white, above left=of red6] (white4) {};  
    
            \draw[very thick] (blue1) -- (white1);
            \draw[very thick] (blue1) -- (red1);
            \draw[very thick] (red2) -- (red1);
            \draw[very thick] (blue1) -- (red3);
            \draw[very thick] (red4) -- (red3);
            \draw[very thick] (red4) -- (white3);
            \draw[very thick] (blue2) -- (red2);
            \draw[very thick] (blue2) -- (white2);
            \draw[very thick] (blue2) -- (red5);
            \draw[very thick] (red6) -- (red5);
            \draw[very thick] (red6) -- (white4);
        \end{scope}
        \end{tikzpicture}
        \caption{An example of the initial coloring and two recoloring steps, with $r=2$ and $\Delta=3$.}\label{fig:recoloring}
    \end{figure}
    
    Notice that after the first rake operation and the compress operation, and then the following $r$ successive rake operations, all white nodes are taken.
    Now, we want to bound the number of leftover blue and red nodes.
    Let $B$ and $R$ be the sets of blue and red nodes.
    Now let $G'$ be the subgraph of $G$ induced by blue and red nodes, rooted at an arbitrary blue node.
    Observe first that every leaf node in $G'$ is a blue vertex.
    Indeed, suppose for sake of contradiction that there is a red leaf node $v$.
    It is red because it was at distance $\leq r$ from a blue node $b$ at the beginning.
    Consider the subtree of $G'$ starting from $b$.
    Because its diameter is at most $r$, this subtree is white at the end of the recoloring.
    This is a contradiction with the fact that $v$ is red.
    Now, we claim the following:
    \begin{align}
        |V(G)| &\geq \Delta |B| + 2, \label{eq:B_small}\\
        |R| &\leq 2r |B|. \label{eq:R_small}
    \end{align}
    Once those claims are proven, we can combine them and end the proof with the following:
    $|V'| = |R| + |B| \leq \O(r)|V(G)|/\Omega(\Delta)$.
    \Cref{eq:B_small} follows from the handshake lemma and the fact that $G$ has $|V(G)|-1$ edges.
    To prove \eqref{eq:R_small}, we first need to define the \emph{depth} of a red node of $G'$ as follows:
    $\depth(v)$ is the minimum distance of $v$ to a blue node by only taking paths of $G'$ that go downward in the tree, i.e.\ the node of the path closest to the root should be $v$.
    By a previous claim, $\depth(v)$ is indeed well-defined for any $v\in R$.
    Said in other words, it would be the minimum distance from $v$ to a blue node if we orient the edges of $G'$ away from the root.

    For any red node $v$, we claim that $\depth(v) \leq 2r$.
    Indeed, assume for sake of contradiction that there is some node $v$ such that all red downward paths in $G'$ of length $2r+1$ starting from $v$ contain no blue node; see \cref{fig:depth-v-2r}.
    Fix one of these paths $P$ and look at the center vertex $u$ of this path. This vertex $u$ splits $P$, we can write $P=P_1 u P_2$, where $v\in P_1$.
    We claim that $u$ is at distance at least $r+1$ from a blue node, which is in contradiction with the fact that $u$ is red.
    To show the claim, assume for sake of contradiction that there exists some downward path $P_3\subseteq V(G')$ of length $r+1$ from $u$ to a blue node $b$.
    Let $P'=P_1 \triangle P_3 \cup \{u\}$ be the path formed by joining $P_3$ to $P_1$, and removing all vertices that appear twice (so that it is indeed a path). This is a path of length $\leq |P_1|+|P_3|\leq r + (r+1)=2r+1$ from $v$ to a $b$, contradiction.
    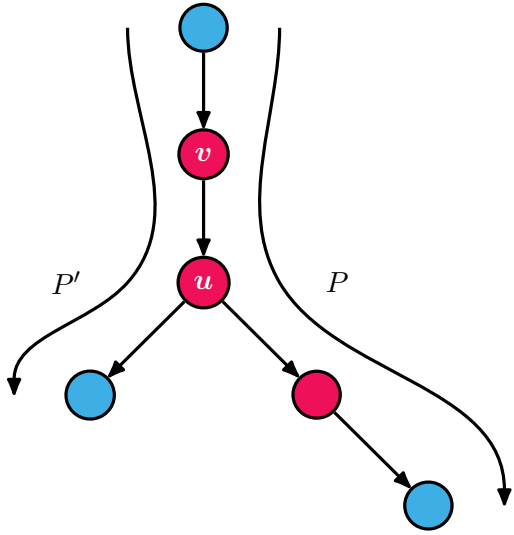
\begin{figure}
        \centering
        \begin{tikzpicture}[every node/.style={draw=black,circle, very thick},use Hobby shortcut]
            \node[fill=OrangeRed] (v) {$\color{white}{\boldsymbol{v}}$};
            \node[fill=CornflowerBlue, above=of v] (blue) {$\phantom{v}$};
            \node[fill=OrangeRed, below=of v] (u) {$\color{white}{\boldsymbol{u}}$};
            \node[fill=CornflowerBlue, below left=of u] (blue2) {$\phantom{u}$};
            \node[fill=OrangeRed, below right=of u] (u1) {$\phantom{v}$};
            \node[fill=CornflowerBlue, below right=of u1] (blue3) {$\phantom{v}$};
    
            \draw[very thick,-{Latex[round]}] ([out angle=270]$(blue.center)-(1,0)$) .. ($(u.center)-(1,0)$) .. ([in angle=90]$(blue2.center)-(1,0)$);
            \draw[very thick,-{Latex[round]}] ([out angle=270]$(blue.center)+(1,0)$) .. ($(u.center)+(1,0)$) .. ([in angle=90]$(blue3.center)+(1,0)$);
            
            \node[draw=none, left=of u] {$P'$};
            \node[draw=none, right=of u] {$P$};
    
            \draw[-{Latex[round]}, very thick] (blue) -- (v);
            \draw[-{Latex[round]}, very thick] (v) -- (u);
            \draw[-{Latex[round]}, very thick] (u) -- (u1);
            \draw[-{Latex[round]}, very thick] (u) -- (blue2);
            \draw[-{Latex[round]}, very thick] (u1) -- (blue3);

        \end{tikzpicture}
        \caption{Explanation of the proof that $\depth(v) \leq 2r$ for $r=1$.}\label{fig:depth-v-2r}
    \end{figure}
    Now, let us count the number of red nodes.
    Each red node receives one token and sends it to one of its nearest blue successor in $G'$ (break ties arbitrarily).
    Because of the depth property, each blue node $b$ can receive only $2r$ tokens at most, from every red node in the red downward path that ends at $b$'s parent.
    Hence, $|R| \leq 2d |B|$, concluding the proof.
\end{proof}

\begin{corollary}\label{cor:mrc_running_time}
    $\ARC(r, \Delta)$ can be computed in $\O(r^2\cdot\log_{\Delta/r} n)$ rounds.
\end{corollary}
\begin{proof}
    By \cref{lem:mrc_shrinkage}, when we apply the rake-and-compress procedure $r$ times, we divide the number of vertices by $\O(\Delta/r)$.
    Therefore, after computing $(r+1)k$ layers, the number of vertices not in layers is divided by $\O({(\Delta/r)}^k)$.
    Therefore, we need to compute at most $\O(r \cdot \log_{\Delta/r} n)$ layers to have all vertices assigned to some layer.
    In total, $\O(r^2\cdot \log_{\Delta/r} n)$ LOCAL rounds are needed, because every layer takes $\O(r)$ LOCAL rounds to compute.
\end{proof}

Here is a key property that is used to build algorithms optimal in $n$, $d$ and $\delta$:
\begin{lemma}\label{lem:raked_neighbors}
    After the $\ARC(r,\Delta-k+1)$ procedure, every relevant degree-$\Delta$ vertex in $G$ is adjacent to $k$ raked vertices (i.e.\ neighbors which were taken due to their degree being $1$). Moreover, we can assume those raked vertices are the $k$ lowest-layer neighbors of $v$.
\end{lemma}
\begin{proof}
    We apply the procedure $\ARC(r,d-k)$ to the input $G$. Denote by $L_1, \dots, L_m$ the resulting layer partition.
    Let $v$ be a degree-$\Delta$ node on layer $i$.
    Consider the set $R$ of the $k$ lowest-layer neighbors of $v$.
    When $u\in R$ is added to its layer, it cannot be due to compress operation, because when at this time, $u$ had at least $\Delta-k$ other neighbors still not assigned to a layer yet, i.e.\ $v$ had residual degree $\geq\Delta-k+1$.
    Therefore, one of its neighbors cannot be added a layer at that time, and we get a contradiction.
    Hence, $u$ gets assigned to a layer because of a rake operation.
\end{proof}

For $r\geq 2$, we have the more general statement:
\begin{lemma}\label{lem:raked_neighbors_dist_2}
    Given a decomposition given by the $\ARC(r,\Delta-k+1)$ procedure with $r\geq 2$, we have the following.
    For every relevant degree-$\Delta$ vertex $v$, and every vertex $u$ in the set of $k$ lowest-layer neighbors of $v$, every neighbor of $u$ different from $v$ is a raked vertex.
    Moreover, $u$ is also raked.
\end{lemma}
\begin{proof}
    We apply the procedure $\ARC(r,d-k)$ to the input $G$, with $r\geq 2$. Denote by $L_1, \dots, L_m$ the resulting layer partition.
    Let $v$ be a relevant degree-$\Delta$ vertex, and let $u$ be a layer-$j$ vertex in the set of $k$ lowest-layer neighbors of $v$.
    Let $w$ be a neighbor of $u$ different from $v$, in layer $i$.
    $u$ must be a raked vertex, because it is one of the $k$ lowest-layer neighbors of $v$.
    Therefore, $w$ is at a lower layer than $u$, i.e.\ $i < j$.
    We also get that $v$ is of degree $\geq d-k$ in $G_j$, as $u$, one of the $k$ lowest-layer neighbors of $w$, is still in $G_j$.
    Therefore, $w$ is at distance $2$ of $v$ in $G_i$, so $w$ must be a raked vertex because $v$ is of degree $\geq d-k$ in $G_i$.
\end{proof}

We will now use these tools to prove upper-bound results for two problems that play a key role in our classification.

\subsection{Bipartite factor problem}

\begin{definition}
    Bipartite $(k,l)$-factor is the problem $\Pi(d,\delta) = (d,\delta,\{k\}, \{l\})$.
\end{definition}

Note that bipartite $(1,1)$-factor is the same problem as bipartite perfect matching.

\begin{lemma}\label{lem:BF}
    When $k$ and $l$ are non-zero constants, bipartite $(k,l)$-factor can be solved in $\O(\log_\delta n)$ rounds.
\end{lemma}
\begin{proof}
    We apply the procedure $\ARC(1,\min\{d-k,\delta-l\})$ on the input graph $G$.
    We therefore have access to \cref{lem:raked_neighbors}.
    Let $L_1, \dots, L_m$ denote the resulting layer partition.
    The general idea is to iterate through the layers in a backwards order, and match a relevant node with some of its raked neighbors (exist by \cref{lem:raked_neighbors}).
    More formally, consider the nodes on layer $i$, and assume that all edges incident to nodes of layers greater than $i$ have labeled.
    Vertices in $L_i$ perform the following steps to label all their incident edges that remain unlabeled:
    \begin{itemize}
        \item Any edge between two nodes in $L_i$ is labeled with \n,
        \item Any white (resp.\ black) node that already has one incident edge labeled {\p} chooses its $k-1$ (resp.\ $l-1$) lowest-layer neighbors, labels the unlabeled edges to these neighbors by {\p}, and labels the rest of the unlabeled incident edges by {\n}.
        If those neighbors do not all exist, then $v$ is not relevant by \cref{lem:raked_neighbors}, and we can label all unlabeled incident edges of $v$ by {\n}.
        \item Any white (resp.\ black) node that already has one incident edge labeled {\p} chooses its $k$ (resp.\ $l$) lowest-layer neighbors, labels the unlabeled edges to these neighbors by {\p}, and labels the rest of the unlabeled incident edges by {\n}.
        If those neighbors do not all exist, then $v$ is not relevant by \cref{lem:raked_neighbors}, and we can label all unlabeled incident edges of $v$ by {\n}.
    \end{itemize}
    Let $M$ be the set of edges labeled with {\p}.
    We have to argue that this case distinction cover all possible cases, that is, every relevant vertex $v$ has at most one higher-layer neighbor $u$ such that $uv\in M$.
    This is true, because if $v$ is a relevant vertex with a higher-layer neighbor $u$ such that $uv\in M$, then by definition of the algorithm $v$ is a raked vertex and cannot have more than one edge to a higher-layer vertex.
    Notice that we can ignore vertices in $L_1$: they are not relevant because they are not of degree $d$ nor $\delta$.
    Furthermore, by definition, every relevant white (resp.\ black) node $v$ is incident $k$ (resp.\ $l$) edges in $M$.
    Therefore, we proved the correctness of the algorithm.
    Additionally, the desired round complexity is achieved, as we run $\ARC(1,\min\{d-k,\delta-l\})$ in $\O(\log_\delta n)$ time by \cref{cor:mrc_running_time} (and because $k,l$ are constants) and go through all the $\O(\log_\delta n)$ layers once in time $\O(\log_\delta n)$.
\end{proof}

\subsection{Quasi-orientation problem}

\begin{definition}
    Quasi-$(k,l)$-orientation is the problem $\Pi(d,\delta) = (d,\delta,\{k\}, \{0,\delta-l\})$.
\end{definition}
This problem is easiest to understand if we interpret it so that the edges in $X$ are oriented towards black nodes and edges in $E \setminus X$ are oriented towards white nodes. Then in a quasi-$(k,l)$-orientation, all relevant white nodes have outdegree $k$ while all relevant black nodes have outdegree $l$ or $\delta$.

\begin{lemma}\label{lem:QO}
    When $k$ is a constant, quasi-$(k,l)$-orientation can be solved in $\O(\log_d n)$ rounds.
\end{lemma}
\begin{proof}
    We apply the procedure $\ARC(2,d-k)$ to the input graph $G$.
    Let $L_1, \dots, L_m$ denote the resulting layer partition.
    The general idea is to iterate through the layers in a backwards order, and orient a white degree-$d$ node towards some of its black raked neighbors, and orient a black degree-$\delta$ node towards some of its white raked neighbors (only if the black node is adjacent to an edge labelled with a {\p}).
    More formally, consider the nodes on layer $i$, and assume that all edges incident to nodes of layers $>i$ have labeled.
    Vertices in $L_i$ perform the following steps to label all their incident edges that remain unlabeled:
    \begin{itemize}
        \item If a black node is not incident to any edge labeled {\p}, label all unlabeled incident edges by~{\n}.
        \item If a black node is incident to exactly one edge labeled {\p}, choose  its $l$ lowest-layer neighbors, label the unlabeled edges to these neighbors by {\n}, and label the rest of the unlabeled incident edges by {\p}.
        These lower-layer neighbors always exist because the black node is a raked vertex. 
        \item If a white node is incident (respectively not incident) to an edge labeled {\p}, choose its $k-1$ (respectively $k$) lowest-layer neighbors, label the unlabeled edges to these neighbors by {\p}, and label the rest of the unlabeled incident edges by {\n}.
        In the case where those neighbors do not all exist, then $v$ is not relevant by \cref{lem:raked_neighbors}, and we can label all unlabeled incident edges of $v$ by {\n}.
    \end{itemize}
    Let $M$ be the set of edges labeled with {\p}.
    We have to argue that this case distinction cover all possible cases, that is, every relevant vertex $v$ has at most one higher-layer neighbor $u$ such that $uv\in M$.
    For relevant black vertices, this is because black vertices with an edge labeled {\p} that goes to a higher-layer neighbor are raked vertices. 
    Now, let us tackle the case of white vertices.
    Suppose a white vertex $u\in L_i$ is adjacent to a higher-layer or equal-layer black vertex $v\in L_j$, with all but $l$ of $v$'s incident edges labeled {\p}.
    By \cref{lem:raked_neighbors_dist_2}, $v$ must be a raked vertex, and one the $k$ lowest-layer neighbors of some relevant node $w\in V(G_j)$.
    We also get that $u$ is a raked vertex.
    Thus, $u$ cannot have 2 incident edges labeled with {\p} coming from higher-layer neighbors.
    Notice that we can ignore white vertices in $L_1$: they are not relevant because they are of degree $<d$.
    Black vertices in $L_1$ have no relevant white neighbor, so we can satisfy them trivially.
    Other relevant black nodes either have $\delta-l$ incident edges labeled {\p} or all their incident edges labeled {\n}, and furthermore, by definition, every white relevant node (thus not in $L_1$) is incident to $k$ edges in $M$.
    Therefore, the algorithm is correct.
    Additionally, the desired round complexity is achieved, as we run $\ARC(2,d-k)$ in $\O(\log_d n)$ time by \cref{cor:mrc_running_time} (and because $k$ and $l$ are constant) and go through all the $\O(\log_d n)$ layers once in time $\O(\log_d n)$.
\end{proof}

\section{Classification of logarithmic problems}\label{sec:log-classification}

We now have all the tools to classify logarithmic problems and prove our main theorem. A summary of the classification is displayed in \cref{fig:log_classification}.

\maintheorem*
\begin{proof}
    Fix some $d$ and $\delta$, let $w=\hat{W}(d)$ and $b = \hat{B}(\delta)$.
    Suppose $\Pi$ is solvable and has $\Theta_{d,\delta}(\log n)$ complexity.
    We know by \cref{lem:log_lower_bound} that the problem has complexity $\Omega(\log_d n)$.
    However, this is not tight: there are some problems that cannot be solved in $\O(\log_d n)$ time.
    We make a case distinction to classify the problem complexity into multiple classes.

    First, if $w$ is center-good, then it can be solved in $\O(\log_d n)$ rounds using \cref{lem:resilient_algo}.
    
    Now suppose that $w$ is edge-good.
    Let $w= s_w w_1 {\n}^k w_2 t_w$ with $|w_1|,|w_2| \leq C$, and $b= s_b \dots t_s$.
    Depending on $w$ and $b$, the complexity splits into different classes.
    The problem is not of constant complexity, so it cannot be that $s_w=s_b=\p$ or $t_w=t_b=\p$.
    If $s_b=t_b=\p$, then $\Pi(d,\delta)$ can be solved with the quasi-$(k,0)$-orientation algorithm of \cref{lem:QO} (and possibly \cref{lem:reverse}) in time $\O(\log_d n)$, because it cannot be that $w_1$ and $w_2$ only contains {\n}'s, otherwise we would have $w={\n}^{d+1}$ (as $s_w=t_w=\n$).
    Else, either $s_b=\n$ or $t_b=\n$.
    Without loss of generality, assume in the following that $t_b=\n$ (otherwise, one can use \cref{lem:reverse}).
    In the case where $s_b=\p$, it is not possible that $b={\p}{\n}^*$, as $s_w=\n$.
    Otherwise, the problem would be unsolvable.
    Therefore, $b$ contains at least $2$ {\p}'s and $\Pi(d,\delta)$ can be solved in $\O(\log_d n)$ rounds with the quasi-$(k,l)$-orientation algorithm.
    
    We are now left with the case $s_b=t_b=\n$ for the rest of the proof.
    Then the problem has complexity $\Omega(\log_\delta n)$ by \cref{prop:w_sparse_lb} (assuming $d=\O(n^{1-\alpha})$ for some $\alpha > 0$).
    Moreover, by assumption, the complexity is $\O(\log n)$.
    There are still two possible complexity classes: $\log_\delta n$ and $\log n$.
    To finish our classification, we need to prove that in all cases, the complexity is either $\Omega(\log n)$ or $\O(\log_\delta n)$.
    If $b$ is center-good, then the complexity is $\O(\log_\delta n)$ by \cref{lem:resilient_algo}.
    For the rest of the proof, suppose $b$ is edge-good.
    Let $b= s_b b_1 {\n}^l b_2 t_b$ with $|b_1|,|b_2| \leq C$.
    If $s_w=t_w=\p$, then $\Pi(d,\delta)$ can be solved in time $\O(\log_\delta n)$ with \cref{lem:switch}, possibly \cref{lem:reverse}, and the quasi-$(k,0)$-orientation algorithm of \cref{lem:QO}. Indeed, we can apply this algorithm because it cannot be that $b_1$ and $b_2$ only contains {\n}'s, otherwise we would have $b={\n}^{\delta+1}$ ($s_b=t_b=\n$).
    There are multiple cases depending on the $w_c$'s and $b_c$'s:
    \begin{itemize}[noitemsep]
        \item Type $A$: there exists $i\in\{1,2\}$ such that both $b_i$ and $w_i$ contain a \p.
        \item Type $B$: $w_1$ and $b_2$ only contain {\n}'s, and $w_2$ and $b_1$ contains a \p.
        \item Type $C$: $w_2$ and $b_1$ only contain {\n}'s, and $w_1$ and $b_2$ contains a \p.
        \item Type $D_c$: there exists $c\in\{w,b\}$ such that both $c_1$ and $c_2$ only contain {\n}'s.
    \end{itemize}
    Without loss of generality, we can forget about type $C$, because by using \cref{lem:reverse} we can get back to type $B$.
    If the problem is of type $A$ then the complexity is $\O(\log_\delta n)$ by \cref{lem:BF} and using \cref{lem:reverse} (if necessary).
    Now suppose the problem is of type $B$.
    If $s_w=\n$, the problem is of complexity $\Omega(\log n)$ by \cref{prop:orientation_lb}. Else, $s_w=\p$ and $t_w=\n$, therefore the problem can be solved in time $\O(\log_\delta n)$ using \cref{lem:QO} and \cref{lem:switch}.
    If the problem is of type $D_b$ or $D_w$, it is unsolvable.
    Suppose that the problem is of type $D_b$. Then the problem is unsolvable because $b={\n}^{\delta+1}$.
    Suppose that the problem is of type $D_w$.
    If $s_w=t_w=\n$, the problem is unsolvable because $w={\n}^{d+1}$.
    Else, one of $s_w$ and $t_w$ is equal to \n, but not both.
    So $w={\n}^+{\n}{\p}$ or $w={\p}{\n}{\n}^+$ and the problem is unsolvable because $s_b=t_b=\n$.
\end{proof}

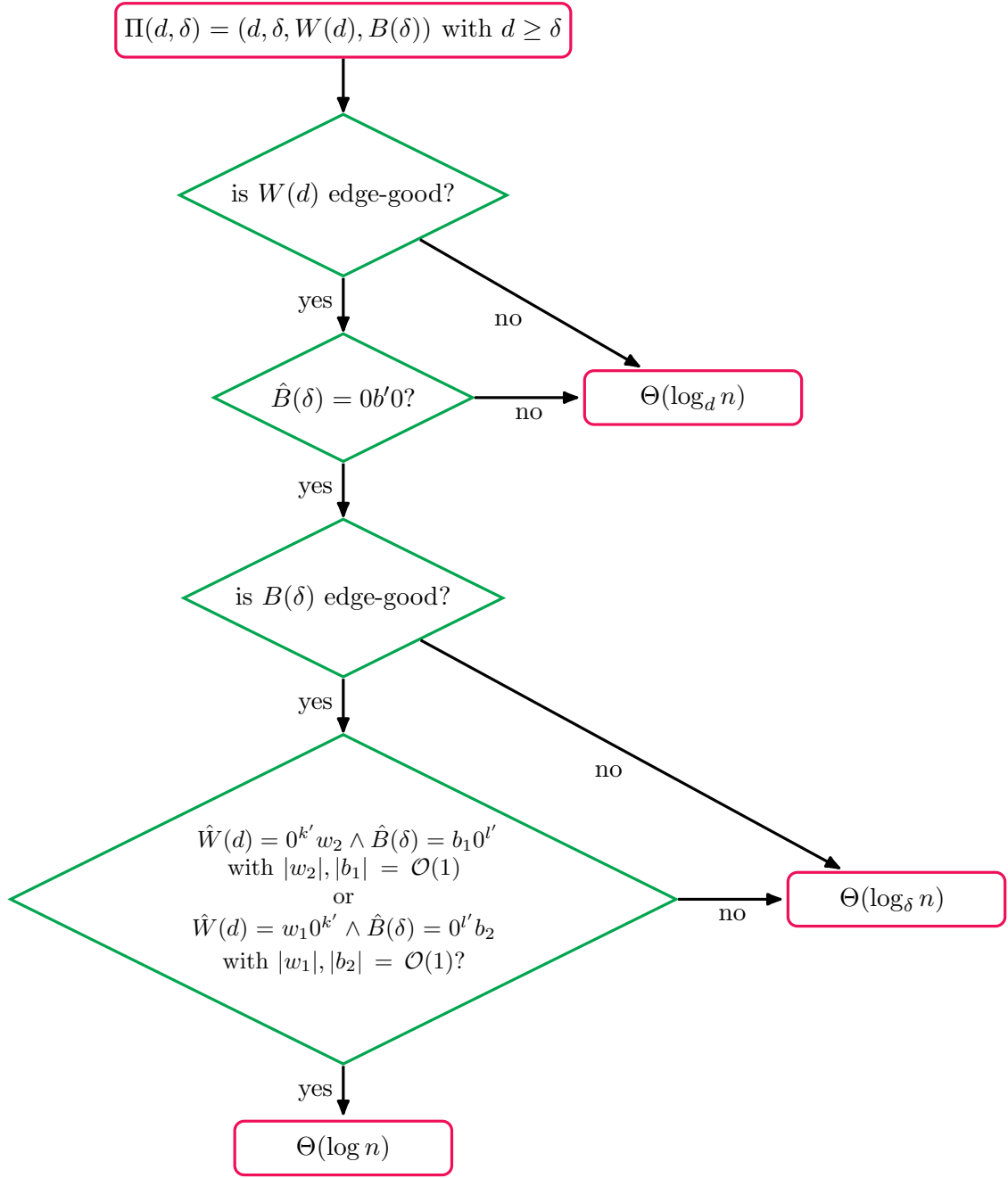
\begin{figure}
    \begin{center}
        \begin{tikzpicture}[
                node distance = 8mm and 16mm,
                base/.style = {very thick,minimum width=32mm, minimum height=8mm, align=center},
                startstop/.style = {base, rounded corners, draw=OrangeRed},
                decision/.style = {base, diamond, aspect=2, draw=Green},
                every edge quotes/.style = {auto=right},
                scale=0.5,
                ]
            \node[startstop] (start) {$\Pi(d,\delta)=(d,\delta,W(d),B(\delta))$ with $d\geq \delta$};
            \node[below=of start, decision] (w_sparse) {is $W(d)$ edge-good?};
            \node[below= of w_sparse,decision] (b_condition) {$\hat{B}(\delta) = \n b' \n$?};
            \node[below= of b_condition,decision] (b_sparse) {is $B(\delta)$ edge-good?};
            \node[below= of b_sparse,text width=12em,decision] (log_condition) 
                {\small{
                    $\hat{W}(d) = {\n}^{k'}w_2\land \hat{B}(\delta) = b_1{\n}^{l'}$ with $|w_2|,|b_1|=\O(1)$\\
                    or\\
                    $\hat{W}(d) = w_1{\n}^{k'}\land \hat{B}(\delta) = {\n}^{l'}b_2$ with $|w_1|,|b_2|=\O(1)$?
                }};
            \node[right=of b_condition, startstop] (log_d_n) {$\Theta(\log_d n)$};
            \node[right=of log_condition, startstop] (log_delta_n) {$\Theta(\log_\delta n)$ };
            \node[below= of log_condition,startstop] (log_n) {$\Theta(\log n)$};

            \draw[very thick, -{Latex[round]}] (start) edge (w_sparse);
            \draw[very thick, -{Latex[round]},shorten >=2pt] (w_sparse) edge["no"] (log_d_n);
            \draw[very thick, -{Latex[round]}] (w_sparse) edge["yes"] (b_condition) ;
            \draw[very thick, -{Latex[round]},shorten >=2pt] (b_condition) edge["no"] (log_d_n);
            \draw[very thick, -{Latex[round]}] (b_condition) edge["yes"] (b_sparse);
            \draw[very thick, -{Latex[round]},shorten >=2pt] (b_sparse) edge["no"] (log_delta_n);
            \draw[very thick, -{Latex[round]}] (b_sparse) edge["yes"] (log_condition);
            \draw[very thick, -{Latex[round]},shorten >=2pt] (log_condition) edge["yes"] (log_n);
            \draw[very thick, -{Latex[round]},shorten >=2pt] (log_condition) edge["no"] (log_delta_n);
        \end{tikzpicture}
    \end{center}
    \caption{Complexity flowchart for logarithmic problems; we assume here that the complexity of $\Pi(d,\delta)$ is $\Theta_{d,\delta}(\log n)$.}
    \label{fig:log_classification}
\end{figure}

\section{Outside logarithmic region}\label{sec:other-classification}

In \cref{sec:log-lower,sec:log-upper,sec:log-classification} we have analyzed problems of complexity $\Theta_{d,\delta}(\log n)$. In this section we explore the two remaining cases of solvable problems: classes $\O_{d,\delta}(1)$ and $\Theta_{d,\delta}(n)$. All results are summarized in \cref{tab:results} on \cpageref{tab:results}.

\subsection{Constant complexity}\label{ssec:constant}

\begin{proposition}\label{prop:constant}
    If the complexity of $\Pi(d,\delta)$ is $\O_{d,\delta}(1)$, then it is also $\O(1)$.
\end{proposition}
\begin{proof}
    If the complexity of the problem is $\O_{d,\delta}(1)$, then by the classification of \cite{balliu2019classification} the problem is of the form $({\p}{\p}{\p}^+, B)$ with $B$ non-empty or of the form $(W,B)$ with $0\in W\cap B$ (or their switch or reverse). To solve the first type of problem in $\O(1)$ rounds, one can run an algorithm where black nodes label $k$ neighboring edges with a \p (if $k\in B$), and the rest of the neighboring edges with a \n.
    To solve the second type of problems in $\O(1)$ rounds, one can run an algorithm where every node labels every neighboring edge with a \n.
    Both of these algorithms are correct and run in $\O(1)$ rounds.
\end{proof}

\subsection{Linear complexity: upper bounds}\label{ssec:linear-upper}

For a bipartite graph $G$ with bipartition $V(G)=V_W\sqcup V_B$, denote $W_s=\{v\in V_W\mid \deg(v)=s\}$ and $B_t=\{v\in V_B\mid \deg(v)=t\}$.
\begin{lemma}\label{lem:diameter_bipartite}
    Let $G$ be a tree with bi-partition $V(G)=V_W\sqcup V_B$.
    Then $G[W_s\cup B_t]$ has components of diameter at most $4|V(G)|/(s+t)+4$.
\end{lemma}
\begin{proof}
    Let $G'$ be a connected component of $G[W_s\cup B_t]$.
    Let $d$ be the diameter of $G'$ and let $P$ be a path of length $d$ in $G'$.
    Notice that in $G'$, all vertices are either leaves, vertices of degree $s$ in $W$, or vertices of degree $t$ in $B$.
    Therefore, $P$ contains at most $2$ leaves, at least $(d-4)/2$ vertices from $W_s$ and at least $(d-4)/2$ vertices from $B_t$.
    By the handshake lemma, $2|V(G)| \geq s|W_s| + t|B_t| \geq (s+t)(d-4)/2$.
    Finally, we get $d\leq 4|V(G)|/(s+t)+4$.
\end{proof}

\begin{corollary}
    Any solvable binary labeling problem $\Pi(d,\delta)$ has complexity $\O(n/(d+\delta))$.
\end{corollary}
\begin{proof}
    Let $(G,d,\delta)$ be an instance of $\Pi$ with $W\sqcup B$ a bipartition of $V(G)$.
    As we do not have only constraints on vertices of $W$ of degree $d$ and on vertices of $B$ of degree $\delta$, one needs only to only solve $\Pi$ on the connected components of $G[W_d\cup B_\delta]$.
    By \cref{lem:diameter_bipartite}, the diameter is at most $4|V(G)|/(d+\delta)+4$, and every LOCAL problem can be solved in a number of rounds equal to the diameter.
    Therefore, $\Pi$ can be solved in time $\O(n/(d+\delta))$.
\end{proof}

\subsection{Linear complexity: lower bounds}\label{ssec:linear-lower}

Up to switching, there are two types of problems with complexity $\Omega_{d,\delta}(n)$ \cite{balliu2019classification}:
\begin{itemize}[noitemsep]
    \item Type A: $(W,B) = ({\p}{\n}^+{\p}, {\n}{\p}{\n})$,
    \item Type B: $(W,B) = (*{\p}{\n}^+, {\n}^+{\p}*)$,
\end{itemize}
where $*$ is a placeholder for an arbitrary bit.
We prove that both types have complexity $\Omega(n/(d+\delta))$.
Notice that in the first type, $\delta=3$ so the complexity is $\Omega(n/d)$.

\begin{lemma}
    If $\Pi(d,3)$ is a problem of type $A$, then $\Pi(d,3)$ has complexity $\Omega(n/d)$.
\end{lemma}
\begin{proof}
    This follows from the proof of Theorem 8.1 of~\cite{balliu2019classification}.
    In this proof, the problem is reduced to $2$-coloring of a path of size $n/d$, and it is known that the complexity of $2$-coloring a path of length $\ell$ is $\Theta(\ell)$. 
    Hence, $\Pi(d,3)$ has complexity $\Omega(n/d)$
\end{proof}

\begin{lemma}
    If $\Pi(d,\delta)$ is of type $B$, then $\Pi(d,\delta)$ has complexity $\Omega(n/(d+\delta))$.
\end{lemma}
\begin{proof}
    This follows the proof of Theorem 8.2 of~\cite{balliu2019classification}.
    In this proof, the problem is reduced to a problem called \emph{almost oriented path} for a path of size $f(n,d,\delta)$, and then it is showed that the complexity of almost orienting a path of length $\ell$ is $\Theta(\ell)$.
    The almost oriented path problem takes as input a properly $2$-colored path and asks to orient the edges such that at most one node of degree $2$ has all its incident edges outgoing, and all other nodes of degree $2$ have exactly one outgoing edge (nodes of degree $1$ are unconstrained).

    Let us compute $f(n,d,\delta)$. In the proof, of Theorem 8.2, the authors start from properly $2$-colored path, connect to each white node $d-2$ new black leaves, and to each black node $\delta-2$ new white leaves.
    Then, the authors connect an additional node to each endpoint of the path while keeping a proper $2$-coloring.
    Then, any solution for $\Pi(d,\delta)$ for this new instance can be mapped to a solution for the almost oriented path problem on the original instance.
    Let $k$ be the order of the original instance.
    $n$ is the size of the new instance so $n \geq (k/2+1)(\delta-2) + (k/2+1)(\delta-2)$ and therefore $k = \Omega(n/(d+\delta))$.
    This finishes the proof that $\Pi(d,\delta)$ has complexity $f(n,d,\delta) = \Omega(k) = \Omega(n/(d+\delta))$.
\end{proof}

\subsection{Linear complexity: summary}

By putting together \cref{ssec:linear-upper,ssec:linear-lower}, we obtain:
\begin{corollary}\label{cor:linear}
    If the complexity of $\Pi(d,\delta)$ is $\Theta_{d,\delta}(n)$, then it is also $\Theta(n/(d+\delta))$.
\end{corollary}

\section{Structurally simple problems and context-free languages}\label{sec:ss-languages}

We have now completed the classification of structurally simple binary labeling problems. In this section, we show that all binary labeling problems that can be defined with context-free grammar are structurally simple:
\sscontextfreelemma*

To prove this, we need the following result from language theory:
\begin{theorem}[{\cite[Theorem 2.1]{ilie1994slender}}]\label{thm:slender}
    Let $B$ a constant and $L$ a context-free language such that for every length $n$ of a word in $L$, there are at most $B$ words of length $n$ in $L$.
    Then $L$ can be written as a finite union of so-called paired loops, i.e.
    \[ L = \bigcup_{i \in I} \{u_i v_i^n w_i x_i^n y_i \mid n\in \N\}\]
    for some finite set $I$ and words $u_i, v_i, w_i, x_i, y_i$ for every $i\in I$.
\end{theorem}

We also need the following auxiliary lemma:
\begin{lemma}\label{lem:language-aux}
    For every language $L = \{u v^n w x^n y \mid n\in \N\}$, for some words $u,v,w,x,y \in \{\n,\p\}^*$, there exists $\alpha > 0$ and $B,C,N\in\N$ such that for $n \geq N$, at least one of these is true:
    \begin{enumerate}[noitemsep]
        \item $L$ contains a word $z$ with length $n$ with a \p at position $i$ s.t.\ $\alpha n -B \leq i \leq (1-\alpha) n + B$,
        \item there is no word of $L$ of length $n$ with a \p at position $i$ s.t.\ $i > C$ or $i < n-C$.
    \end{enumerate}
\end{lemma}
\begin{proof}
    Let $N = |uvwxy|$.
    Let $C = \max\{|u|,|v|,|w|,|x|,|y|\}$.
    If $v, w, x \in \n^*$, all the words $z\in L$ of length $n\geq N$ can only contain symbol \p's at positions $i\in [0,C]\cup[|z|-C,|z|]$.
    Otherwise, we know that the word $vwx$ contain a symbol \p. 
    Let $n\geq N$. Take $z\in L$ of length $n$ and let $k\in\N$ such that $z = u v^k w x^k y = u v^{k-1} (vwx) x^{k-1} y$.
    $v^{k-1}$ and $x^{k-1}$ exist, because $n \geq N$.
    Therefore, there is a \p in $z$, inside the middle $vwx$, at position \[i\in[|u|+(k-1)|v|, |z|-|y|-(k-1)|x|].\]
    There are multiple cases.
    We first handle the case where $v$ and $x$ are not empty.
    If this is the case, then let $\alpha = \min\{|v|, |x|\}/(|v| + |x|)$ and $B=\alpha|uvwxy|$.
    Then $\alpha n - B = \alpha((k-1)|v|+(k-1)|x|) = (k-1)\alpha(|v|+|x|) \leq (k-1)|v|, (k-1)|x|$.
    This means that $i, |z|-i \geq \alpha n - B$.
    Let us now handle the case where $v$ is empty; the proof when $x$ is empty is similar, and $v$ and $x$ cannot be both empty.
    Let $C = 2\max\{|u|,|v|,|w|,|x|,|y|\}$, $\alpha = 1/2$, and $B = |uvwxy|$.
    If $x\in \n^*$, then $z$ only contain symbol \p's at positions $i\in [0,C]\cup[|z|-C,|z|]$.
    Otherwise, write $z=u w x^{\lfloor k/2 \rfloor} x x^{\lceil k/2 \rceil - 1}y$. 
    Notice that there is a \p in $z$, inside the middle $x$, at position \[i\in\bigl[|u|+|w|+\lfloor k/2 \rfloor|x|, |z|-|y|-(\lceil k/2 \rceil - 1)|x|\bigr].\]
    Then $\alpha n - B \leq (k-2)|x|/2 \leq \lfloor k/2 \rfloor|x|, (\lceil k/2 \rceil - 1)|x|$.
    This means that $i, |z|-i \geq \alpha n - B$.
    This concludes the proof.
\end{proof}

\begin{proof}[Proof of \cref{lem:ss-context-free}]
    Using \cref{thm:slender}, we can write
    \[
        \hat X = \bigcup_{i \in I} \{u_i v_i^k w_i x_i^k y_i \mid k\in\N\}
    \]
    for some finite set $I$ and words $u_i$, $v_i$, $w_i$, $x_i$, $y_i$ for every $i\in I$.
    Apply \cref{lem:language-aux} for each $L_i = \{u_i v_i^k w_i x_i^k y_i \mid k\in\N\}$ separately, and obtain constants $\alpha_i, B_i, C_i$, and $N_i$ for each $i \in I$, so that for every $i\in I$, $n\geq N_i$, and $w\in L_i$ of length $n\geq N_i$ at least one of these is true:
    \begin{enumerate}[noitemsep]
        \item $L_i$ contains a word $w$ with length $n$ with a \p at position $x$ s.t.\ $\alpha_i n -B_i \leq x \leq (1-\alpha_i) n + B_i$,
        \item there is no word of $L$ of length $n$ with a \p at position $x$ s.t.\ $x > C_i$ or $x < n-C_i$.
    \end{enumerate}
    Let $\alpha = \min\{\alpha_i \mid i \in I\}$, $B = \max\{B_i \mid i \in I\}$, $C = \max\{C_i \mid i \in I\}$, and $N = \max\{N_i \mid i \in I\}$.
    Then, given some $n\geq N$, at least one of these is true:
    \begin{enumerate}[noitemsep]
        \item $\hat X$ contains a word $w\in L_i$ for some $i\in I$ with length $n$ with a \p at position $x$ s.t.\ $\alpha n - B \leq x \leq (1-\alpha) n + B$, because $\alpha n - B \leq \alpha_i n -B_i$ and $(1-\alpha_i) n + B_i \leq (1-\alpha) n + B$,
        \item there is no word of $\hat X$ of length $n$ with a \p at position $x$ s.t.\ $x > C$ or $x < n-C$, because for every $i\in I$, $C \geq C_i$ or $n-C \leq n-C_i$.
    \end{enumerate}
    This proves that for any $\eps > 0$, there exists some $N$ such that for any word~$w\in \hat X$ with length~$n\geq N$, at least one of these is true:
    \begin{enumerate}[noitemsep]
        \item word~$w$ is $\eps$-center-good: it has a \p at position $p \in [n^\eps,n^{1-\eps}]$,
        \item word~$w$ is $C$-edge-good: it has \p's only at positions $p\in [0,C]\cup[|w|-C,|w|]$.
    \end{enumerate}
    This means that $X$ is structurally simple. 
\end{proof}

\section*{Acknowledgements}

This work was supported in part by the Research Council of Finland, Grant 333837, and it was partly done while TP was doing an internship at Aalto University.

\bibliography{bibliography.bib}

\end{document}